\documentclass[letterpaper, 10 pt, conference]{ieeeconf}  

\IEEEoverridecommandlockouts                              
\usepackage{graphics} 
\usepackage{epsfig} 
\usepackage{times} 
\usepackage{amsmath} 
\usepackage{amssymb, bm}  
\usepackage[capitalise]{cleveref}
\usepackage{xcolor}
\usepackage{algorithm}
\usepackage{algpseudocode}
\usepackage{subfigure,color}
\newtheorem{rmk}{Remark}
\newtheorem{assump}{Assumption}

\newtheorem{definition}{Definition}
\newtheorem{thm}{Theorem}
\newtheorem{lemma}{Lemma}

\title{\LARGE \bf
	Learning-based Parameterized Barrier Function for Safety-Critical Control of Unknown Systems
}

\author{
    Sihua Zhang$^{1}$, Di-Hua Zhai$^{1*}$, Xiaobing Dai$^{2}$, Tzu-yuan Huang$^{2}$, Yuanqing Xia$^{1}$, Sandra Hirche$^{2}$
    \thanks{$^*$ Corresponding author.}
	\thanks{This work was supported in part by the National Natural Science
		Foundation of China under Grant 62173035, Grant 61803033 and Grant 61836001, in part by the European Union’s Horizon Europe innovation action programme under grant agreement No. 101093822, "SeaClear2.0", and by the DAAD programme Konrad Zuse Schools of Excellence in Artificial Intelligence. 
    }
	\thanks{Sihua Zhang, Di-Hua Zhai and Yuanqing Xia are with School of Automation, Beijing Institute of Technology, Beijing, China. \{sihua.zhang, zhaidih, xia\_yuanqing\}@bit.edu.cn.}%
	\thanks{Xiaobing Dai, Tzu-yuan Huang and Sandra Hirche are with the Chair of Information-oriented Control, TUM School of Computation, Information and Technology, Technical University of Munich, Munich, Germany. \{xiaobing.dai, tzu-yuan.huang, hirche\}@tum.de.}%
}

\begin{document}

	\maketitle
	\thispagestyle{empty}
	\pagestyle{empty}

	\begin{abstract}
		With the increasing complexity of real-world systems and varying environmental uncertainties, it is difficult to build an accurate dynamic model, which poses challenges especially for safety-critical control. 
        In this paper, a learning-based control policy is proposed to ensure the safety of systems with unknown disturbances through control barrier functions (CBFs).
        First, the disturbance is predicted by Gaussian process (GP) regression, whose prediction performance is guaranteed by a deterministic error bound.
        Then, a novel control strategy using GP-based parameterized high-order control barrier functions (GP-P-HOCBFs) is proposed via a shrunk original safe set based on the prediction error bound. 
        In comparison to existing methods that involve adding strict robust safety terms to the HOCBF condition, the proposed method offers more flexibility to deal with the conservatism and the feasibility of solving quadratic problems within the CBF framework.
        Finally, the effectiveness of the proposed method is demonstrated by simulations on Franka Emika manipulator.
	\end{abstract}
	
	
	\section{INTRODUCTION}
	Safety is an essential requirement for real-world systems operating in complex environments, such as collision avoidance in multi-robot environments and secure interaction between human operators and robots \cite{yang2020neural}. 
    Various approaches have been developed to ensure the safety of dynamic systems, including artificial potential field methods \cite{singletary2021comparative} and model predictive control (MPC) \cite{mayne2000constrained}. 
    Moreover, when the formal definition of safety is proposed via the forward invariance of state sets for dynamical systems \cite{blanchini1999set}, control barrier functions (CBFs) \cite{prajna2004safety} become a popular tool for synthesizing safe controllers for dynamical systems. Due to its computational efficiency, CBF theory has gained widespread adoption in safety control applications \cite{xu2017realizing, lindemann2019control}.
    For more general safe constraints, i.e., with arbitrary relative degrees, high-order CBFs (HOCBF) are proposed \cite{xiao2021high} ensuring safety during the entire operation with available precise models of the dynamical systems.
    However, the availability of the system models is practically hard, especially with high model complexity and environmental disturbance. 

    For uncertain system models, various CBF-based methods have been proposed to guarantee safety. A common technique for systems with external disturbance to achieve real-time safety control is robust CBF (rCBF) \cite{nguyen2021robust,gurriet2018towards}. In this approach, the disturbance is bounded to an $\varepsilon$-ball in some norm form with a known boundary. A compensation term is then added to the CBF condition based on the bound of uncertainty to counteract the effects of uncertainty. However, obtaining the disturbance bound can be challenging. Therefore, different identification techniques, such as disturbance observer and extended state observer, have been employed to mitigate the conservatism of robust CBF by estimating the uncertainty \cite{agrawal2022safe, dacs2022robust}. Moreover, 
    data-driven methods are another technique for predicting unknown dynamics, where Gaussian processes (GP) are widely employed in safety-critical scenarios due to the modeling flexibility and the rigorous prediction performance quantification.
    The combination of GP and rCBF is widely studied \cite{emam2021data, wang2018safe}, where the prediction error is fully compensated during the design of rCBF-based quadratic programming (QP).
    While the GP-rCBF ensures the safety of unknown systems, the induced large compensation brings more conservatism and challenges to the feasibility of solving the derived QP problem.
    \looseness=-1
    	
    In contrast to rCBF, input-to-state safety CBF (ISSf-CBF) \cite{kolathaya2018input} partially compensates the prediction errors for better feasibility.
    To deal with the state-dependent conservatism brought by ISSf-CBF, a tunable ISSf-CBF is proposed in \cite{alan2021safe} to achieve the balance between safety and robustness based on the distance to the safe boundary.
    Furthermore, the idea of ISSf-CBF is extended to ISSf-HOCBF \cite{jiang2023safety} to deal with more general systems without the knowledge of disturbances or prediction errors.
    However, due to the lack of full compensation, ISSf-HOCBF results in a potential violation of the safety conditions, whose violation level depends on the effects of uncertainties.
    A promising way to compensate for the violation under ISSf-CBF is to design a new safe set characterized by parameterized barrier functions (PBF) \cite{alan2023parameterized}, such that the violation in the new set still ensures the safety.
    While the PBF framework includes the rCBF and ISSf-CBF, the consideration of the conservatism reduction and the extension to more general safety constraints with high relative degrees have not been addressed yet. 

    Inspired by PBF, HOCBF and tunable ISSf-CBF, this paper proposes a flexible framework, namely GP-based parameterized HOCBF (GP-P-HOCBF), to ensure the safe control of unknown systems under general safety constraints.
    The unknown dynamics is estimated by GP regression, where the prediction is employed in the QP with HOCBF to generate the safe control input.
    Due to the existence of prediction error, a new set of barrier functions is designed based on the theoretical GP prediction error bound, forming a smaller safe set than original.
    Moreover, for more flexibility in feasibility improvement and conservatism reduction, a ISSf-CBF like structure is applied on the newly designed barrier functions with a tunable function, such that the original safety constraint is guaranteed.
    It is shown that the choice of the larger tunable function enhances the feasibility of QP problem but induces a larger shrink from the original safe set.
    Moreover, for the same feasibility requirement, the proposed method provides flexibility to adjust the admissible safe set through distributing the shrink over different orders of CBFs. 
    Finally, the effectiveness of the proposed control strategy is demonstrated through simulations on a Franka Emika manipulator.
    \looseness=-1
	\section{PRELIMINARIES}
	In this section, some basic concepts of the control barrier function (CBF), and high-order CBF are recalled.
	\subsection{Control barrier function}
	Consider an affine control system
    \begin{equation} \label{dynamic-disturbance}
		\dot{\bm{x}} = \bm{f}(\bm{x}) + \bm{g}(\bm{x}) \bm{u} + \bm{d}( \bm{x} ),
	\end{equation}
	where $\bm{x} \in \mathbb{X} \subseteq \mathbb{R}^n$ and $\bm{u} \in \mathbb{U} \subseteq \mathbb{R}^s$ are the system states and control input, respectively.
	The drift function $\bm{f}(\cdot): \mathbb{X} \to \mathbb{R}^n$ and the input gain function $\bm{g}(\cdot): \mathbb{X} \to \mathbb{R}^{n \times s}$ are locally Lipschitz continuous functions. 
	The continuous function $\bm{d}(\cdot) = [d_1(\cdot), \dots, d_n(\cdot)]^T: \mathbb{X} \to \mathbb{R}^n$ denotes the model uncertainties and environmental disturbances. For any initial state $\bm{x}(t_0)$, $\bm{x}(t)$ is the unique solution to system \eqref{dynamic-disturbance}.
	
	For safety-critical control, a closed set $C$ is defined by a continuous differentiable function $h(\bm{x}): \mathbb{R}^n \to \mathbb{R}$ as
	\begin{align}\label{C}
		C=\{ \bm{x}\in \mathbb{R}^n: h(\bm{x})\geq 0 \},
	\end{align}
	It is assumed that $C$ is nonempty and has no isolated point. If for every $\bm{x}(t_0) \in C$, the state $\bm{x}(t)$ always stays in the set $C$ for $t\geq t_0$, the set $C$ is forward invariant \cite{blanchini1999set} and the safety of system \eqref{dynamic-disturbance} is guaranteed. The set $C$ is called the safe set.
    \looseness=-1
	
	To ensure the forward invariance of set $C$, the control barrier function (CBF) is proposed with some important definitions introduced below.
	
	\begin{definition}[Class $\mathcal{K}$ Function \cite{hassan2002nonlinear}]
		A continuous function $\alpha: \left[ 0, a\right) \to \left[ 0, \infty \right), a>0$, is a class $\mathcal{K}$ function if it is strictly increasing and $\alpha(0)=0$.
	\end{definition}
	
	\begin{definition}[Relative Degree \cite{hassan2002nonlinear}]
		For a continuous differentiable function $h(\bm{x}): \mathbb{R}^n \to \mathbb{R}$ with respect to system \eqref{dynamic-disturbance}, the relative degree is the number of times it needs to be differentiated along with its dynamics until the control input $\bm{u}$ explicitly shows in the corresponding derivative.
	\end{definition}
	
	When the relative degree of function $h(\bm{x})$ is $m \!\in\! \mathbb{N}_{>0}$, and the inequality $h(\bm{x}) \!\geq\! 0$ is used as a safety constraint, the definition of CBF is given for $m \!=\! 1$ as follows.
	\begin{definition}[Control Barrier Function]
		Given a set $C$ as in (\ref{C}), $h(\bm{x})$ is a control barrier function (CBF) for system (\ref{dynamic-disturbance}) if there exists a class $\mathcal{K}$ function $\alpha(\cdot)$ such that
		\begin{equation} \label{cbf}
			\sup_{\bm{u} \in \mathbb{U}} \Big\{\! \frac{\partial h(\bm{x})}{\partial \bm{x}} (\bm{f}(\bm{x}) \!+\! \bm{g}(\bm{x})\bm{u} + \bm{d}( \bm{x})) \!+\! \alpha(h(\bm{x})) \!\Big\} \!\geq\! 0,
		\end{equation}  
	\end{definition}
	
	The forward invariance of $C$ is guaranteed as follows.
	\begin{lemma}[\cite{jankovic2018robust}]
		\label{lemma: cbf}
		Given the set $C$ defined by \eqref{C} for a continuous differentiable function $h(\bm{x})$, if $h(\bm{x})$ is a CBF, then Lipschitz continuous control input $\bm{u}(t) \in K_{cbf}(\bm{x}) =  \{ \bm{u} \in \mathbb{U}: \frac{\partial h(\bm{x})}{\partial \bm{x}} (\bm{f}(\bm{x})+\bm{g}(\bm{x})\bm{u} + \bm{d}( \bm{x})) + \alpha(h(\bm{x})) \geq 0 \}$ renders the set $C$ forward invariant.
	\end{lemma}

    For general cases with arbitrary relative degree $m \in \mathbb{N}_{>0}$, similar definitions and conclusions are shown in the following subsection.
    
	\subsection{High order control barrier function}
    \label{subsection_HOCBF}
	When the relative degree of $h(\bm{x})$ satisfies $m>1$, the term $\frac{\partial h(\bm{x})}{\partial \bm{x}}\bm{g}(\bm{x}) \!=\! 0$. In this case, the CBF can not be used to guarantee the safety of system \eqref{dynamic-disturbance}. Therefore, the high order control barrier function (HOCBF) is proposed, where firstly a sequence of functions $\psi_i(\bm{x})\!: \mathbb{R}^n \!\to\! \mathbb{R}$, $i \!\in\! \{ 0,\dots,m\}$ are defined as
	\begin{equation}
		\begin{aligned}
			\label{psi}
			&\psi_0(\bm{x})=h(\bm{x}), \\
			&\psi_i(\bm{x})=\dot{\psi}_{i-1}(\bm{x})+\alpha_i(\psi_{i-1}(\bm{x})),
		\end{aligned}
	\end{equation}
	in which $\alpha_i(\cdot)$ denotes $(m-i)^{th}$ order differentiable class $\mathcal{K}$ function.
    The corresponding safe sets are defined as
    \begin{align}
        C_i=\{\bm{x}\in \mathbb{X}: \psi_{i-1}(\bm{x}) \geq 0\}
    \end{align}
	for $i\in \{1,\dots,m\}$.	
	Given the functions $\psi_i(\bm{x}): \mathbb{R}^n \to \mathbb{R}$, $i\in \{0,\dots,m\}$, the definition of high order control barrier function (HOCBF) is shown as below.
	\begin{definition}[HOCBF \cite{xiao2021high}]
		\label{define-hocbf}
		A function $h(\bm{x}): \mathbb{X} \to \mathbb{R}$ is a HOCBF of relative degree $m$ for system (\ref{dynamic-disturbance}), if there exist $(m-i)^{th}$ order differentiable class $\mathcal{K}$ functions $\alpha_i, i\in \{1,\dots,m-1\}$, and a class $\mathcal{K}$ function $\alpha_m$ such that		
		\begin{align}
			\label{hocbf}
			\sup\nolimits_{\bm{u} \in \mathbb{U}} \{\dot{\psi}_{m-1}(\bm{x})+\alpha_{m} (\psi_{m-1}(\bm{x})) \} \geq 0,
		\end{align}
		for all $\bm{x} \in C_1\cap\dots\cap C_m$.
	\end{definition}
	
	Similar to \cref{lemma: cbf}, the following result also guarantees the forward invariance of set $C_1\cap\dots\cap C_m$.
	\begin{lemma}[\cite{xiao2021high}]
		\label{lemma: hocbf}
		The set $C_1\cap\dots\cap C_m$ is forward invariant for system (\ref{dynamic-disturbance}) if $\bm{x}(0) \in C_1\cap\dots\cap C_m$ and $h(\bm{x})$ is a HOCBF.
	\end{lemma}

    With the above definitions and results, we introduce the problem setting of this paper in \cref{section_problem_setting}.
    
	\section{PROBLEM STATEMENT}\label{section_problem_setting}
	The concepts of CBF and HOCBF both are based on the precise dynamic model (\ref{dynamic-disturbance}), but the external disturbance usually is difficult to get. In this section, some assumptions on the disturbance are given, and the goal of this paper is formulated.   
	
	For the system \eqref{dynamic-disturbance}, the drift function $\bm{f}(\bm{x})$ and the input gain function $\bm{g}(\bm{x})$ are assumed to be known, but the disturbance $\bm{d}(\bm{x})$ is unknown and it satisfies the following assumption.
	
	\begin{assump} \label{assump1}
		Given a kernel function $k_i(\cdot,\cdot): \mathbb{X} \times \mathbb{X} \to \mathbb{R}_{0,+}$ for $\forall i = 1, \dots, n$, the unknown function $d_i(\cdot)$ belongs to the reproducing kernel Hilbert space (RKHS) $\mathcal{H}_{k_i}$ corresponding to $k_i$ with the bounded RKHS norm by a well-defined known constant $B_i \in \mathbb{R}_{0,+}$, i.e., $\| d_i \|_{k_i} \leq B_i$.
	\end{assump}
	\cref{assump1} defines a potential function space of the unknown function $d_i(\cdot)$, which is only related to $k_i(\bm{x}, \cdot), \forall \bm{x} \in \mathbb{X}$.
	The kernel function reflects the correlation between $d_i(\cdot)$ at two points, which requires less prior information than the structure of parametric models e.g., neural networks.
	Choosing universal kernels, e.g., square exponential kernel, the RKHS $\mathcal{H}_{k_i}$ includes all continuous functions, providing the property of universal approximation \cite{micchelli2006universal}.
	The RKHS norm $\| d_i \|_{k_i}$ represents the smoothness of $d_i$, such that the existence of the upper bounds $B_i$ is regarded as the requirement of Lipschitz continuity.
	The value of $B_i$ can be approximated with data-driven method \cite{scharnhorst2022robust}, and therefore \cref{assump1} imposes no practical restrictions.
	
	To obtain $d_i$ from the function space $\mathcal{H}_{k_i}$, the data set $\mathcal{D} = \{ \bm{x}^{(\iota)}, \bm{y}^{(\iota)} \}_{\iota = 1, \dots, M}$ with $M \in \mathbb{N}_{>0}$ is required satisfying the following assumption.
	\begin{assump} \label{assump2}
		The data pair $\{ \bm{x}, \bm{y} \}$ is available with $\bm{y} = [y_1, \dots, y_n]^T = \bm{d}(\bm{x}) + \bm{v}$, where $\bm{v}$ denotes the measurement noise satisfying $\| \bm{v} \|_\infty \leq \sigma_v$ and $\sigma_v \in \mathbb{R}_{0,+}$.
	\end{assump}
	
	\cref{assump2} requires the accessibility of the system state $\bm{x}$, which is commonly found in most advanced control laws, including back-stepping and CBF-based controller.
	Moreover, it also allows noisy measurement of $\bm{d}(\cdot)$, which can be calculated by $\bm{y} \!=\! \dot{\bm{x}} \!-\! \bm{f}(\bm{x}) \!-\! \bm{g}(\bm{x}) \bm{u}$ with $\dot{\bm{x}}$ approximated by numerical methods from $\bm{x}$, e.g., finite difference.
	The upper bound $\sigma_v$ can be easily obtained in practice, inducing the noise variance smaller than $\sigma_v^2$ and making \cref{assump2} not restrictive.
	
	The control objective is to maintain the system state $\bm{x}$ in a safe set $C$ as \eqref{C} defined by a safety constraint function $h(\cdot): \mathbb{X} \!\to\! \mathbb{R}$, which is at least $m$-th order differentiable with relative degree $m$.
    To ensure safety with model uncertainties, a learning-based control policy with HOCBF is proposed, where machine learning methods are applied to predict the uncertainty $\bm{d}(\cdot)$.
    Moreover, a smaller safety set $C^* \!\subseteq\! C$ with new function $h^*(\cdot): \mathbb{X} \!\to\! \mathbb{R}$ is determined to address the effect of prediction error from machine learning.
	
	\section{MAIN RESULT}
	In this section, a GP regression approach to learn the disturbance $d_i$ with a deterministic error bound is provided. Based on the error bound, a new learning-based HOCBF is formulated to guarantee the safety of system (\ref{dynamic-disturbance}). 
	
	\subsection{Gaussian Process}
    Gaussian process, as a kernel method, induces a Gaussian distribution over functions defined by the kernel function $k_i(\cdot,\cdot)$ as $d_i \sim \mathcal{GP}(0, k_i)$ and $d_i \in \mathcal{H}_{k_i}$.
	By combining the data set $\mathcal{D}$, which satisfies $|\mathcal{D}| = M$, with \cref{assump2} and applying the Bayesian principle, the prediction of $\bm{d}(\bm{x})$ at $\bm{x} \in \mathbb{X}$ as a Gaussian distribution with posterior mean $\bm{\mu}(\bm{x}) [\mu_1(\bm{x}), \dots, \mu_n(\bm{x})]^T$ and variance $\bm{\Sigma}(\bm{x}) = \mathrm{diag}(\sigma_1^2(\bm{x}), \dots, \sigma_n^2(\bm{x}))$ as
	\begin{align}
		&\mu_i(\bm{x}) = \bm{k}^{T}_{\mathcal{D},i}(\bm{x}) ( \bm{K}_{\mathcal{D},i} + \sigma_v^2 \bm{I}_{M})^{-1} \bm{y}_{\mathcal{D},i}, \\
		&\sigma_i^2(\bm{x}) = k_i(\bm{x},\bm{x}) - \bm{k}^{T}_{\mathcal{D},i}(\bm{x}) ( \bm{K}_{\mathcal{D},i} + \sigma_v^2 \bm{I}_{M})^{-1} \bm{k}_{\mathcal{D},i}(\bm{x}), \nonumber
	\end{align}
	where $\bm{k}_{\mathcal{D},i}(\bm{x}) = [k_i(\bm{x}, \bm{x}^{(1)}), \dots, k_i(\bm{x}, \bm{x}^{(M)})]^T$, $\bm{K}_{\mathcal{D},i} = [k_i(\bm{x}^{(i)}, \bm{x}^{(j)})]_{i,j  = 1, \dots, M}$ and $\bm{y}_{\mathcal{D},i} = [y^{(1)}_i, \dots, y^{(M)}_i]^T$.
	While the posterior mean $\bm{\mu}(\cdot)$ serves as the compensation of $\bm{d}(\cdot)$, the posterior variance $\bm{\Sigma}(\cdot)$ is used to quantify the prediction performance, which is shown as follows.
    \begin{lemma}[\cite{hashimoto2022learning}]
		\label{lemma: error-bound}
        Consider an unknown function $\bm{d}(\cdot)$ satisfying \cref{assump1}, which is predicted by GP regression using the training data set $\mathcal{D}$ with $| \mathcal{D} | = M$ and \cref{assump2}.
        Then, the prediction error is uniformly bounded as
        \begin{align} \label{eqn_GP_bound}
            \| \bm{\mu}(\bm{x}) - \bm{d}(\bm{x}) \| \le \eta(\bm{x}) = \sqrt{ \mathrm{tr} (\bm{\mathcal{B}} \bm{\Sigma}(\bm{x}) )},
        \end{align}
        for $\bm{x} \in \mathbb{X}$, where $\bm{\mathcal{B}} = \mathrm{diag}(\beta_1, \dots, \beta_n)$ with
        \begin{align}
            \beta_i = B_i^2 - \bm{y}_{\mathcal{D},i}^T (\bm{K}_{\mathcal{D},i} + \sigma_v^2 \bm{I}_M)^{-1} \bm{y}_{\mathcal{D},i} + M
        \end{align}
        and $B_i$ from \cref{assump1}.
	\end{lemma} 

    Compared to the probabilistic error bounded provided in \cite{berkenkamp2016safe}, \cref{lemma: error-bound} shows a deterministic bound, which is more suitable in safety-critical scenarios despite more conservatism, such that the violation of the prediction bound is zero.
    Note that the decreasing trend of $\eta(\cdot)$ for increasing number of training data samples $M$ is proven in \cite{hashimoto2022learning}, making the decrease of the conservatism of control performance possible. However, the induced larger computational time due to the complexity $\mathcal{O}(N)$ and $\mathcal{O}(N^2)$ of $\bm{\mu}(\cdot)$ and $\bm{\Sigma}(\cdot)$ can deteriorate the prediction and thus control performance.
    Therefore, the selection of training data number should strike a balance between the performance conservatism and computation efficiency according to the actual situation.
    
    According to \cref{lemma: error-bound} and given data set $\mathcal{D}$, it is direct to derive the state independent prediction error bound as
    \begin{align}\label{eq: error-bound}
        \| \bm{\mu}(\bm{x}) - \bm{d}(\bm{x}) \| \le \bar{\eta}_{\mathcal{D}} \le \bar{\eta},
    \end{align}
    where $\bar{\eta}_{\mathcal{D}} \!=\! \max_{\bm{x} \in \mathbb{X}} \eta(\bm{x})$ is related to the given data set $\mathcal{D}$.
    The data set independent bound $\bar{\eta} \!=\! \sqrt{ \sum_{i=1}^n (B_i^2 \!+\! M) \max_{\bm{x} \in \mathbb{X}} k_i(\bm{x}, \bm{x}) }$ only requires the size of the data set $M$ and the kernels $k_i(\cdot,\cdot)$ for $i \!=\! 1,\!\dots\!,n$, which is derived due to the positive definite $\bm{K}_{\mathcal{D},i} \!+\! \sigma_v^2 \bm{I}_M$ for any $\mathcal{D}$.
    \looseness=-1

    Based on the bounded prediction error $\bar{\eta}$, a new safety constraint, i.e., new barrier function, is constructed, resulting in a new safe set.
    Next, a method is introduced for newly designed barrier function, such that the original safe set maintains forward invariant and the original safety is preserved.
    \looseness=-1
	
	\subsection{GP-based parameterized HOCBF}
	
	To avoid unsafe behavior caused by inaccurate uncertainty prediction, the HOCBF-based safety guarantees is extended to GP-based parameterized HOCBF by introducing a parameterized function $h^*(\cdot)$ related with prediction error bound $\bar{\eta}$ from $h(\cdot)$.
    Similarly as the HOCBF in \cref{subsection_HOCBF}, a sequence of functions $\psi^*_i(\bm{x}): \mathbb{X} \to \mathbb{R}$, $i \in \{0,\dots,m-1\}$ are defined based on original functions $\psi_i(\bm{x}): \mathbb{X} \to \mathbb{R}$, $i\in \{0,\dots,m-1\}$ for HOCBF as
	\begin{align}
		\label{eq: psi_0}
		\psi^*_i(\bm{x})
        &\!=\! \psi_i(\bm{x}) \!-\! \gamma_{i+1}(\psi_i(\bm{x}), \bar{\eta}^2), 
	\end{align}
	where $\gamma_i(\psi_{i-1}, \bar{\eta}^2): \mathbb{R} \times \mathbb{R}_{0, +} \to \mathbb{R}_{0, +}$ is a continuous differentiable function w.r.t $\psi_{i-1}$, In addition, there exist $(m-i)^{th}$ order differentiable extended class $\mathcal{K}$ function $\alpha^*_i(\cdot)$ such that the function $\psi^*_i(\bm{x})$ satisfies
    \begin{align} \label{eqn_psi_star_derivative_definition}
        \psi^*_i(\bm{x})& \!\le\! \dot{\psi}^*_{i-1}(\bm{x}) \!+\! \alpha^*_i(\psi^*_{i-1}(\bm{x})) + \Delta_{i}(\bm{x}),
    \end{align}
    where $\Delta_{i}(\bm{x}) \!=\! \max\{ 0, \gamma_{i\!+\!1}(\psi_{i}(\bm{x}), \bar{\eta}^2) \!-\! \gamma_{i+1}(\psi_{i\!-\!1}(\bm{x}), \bar{\eta}^2) \}$. 
    With the introduced functions $\psi^*_i$, the corresponding new safe sets $C^*_i$ are defined as
	\begin{align}
		\label{c_delta}
		C^*_i=\{ \bm{x} \in \mathbb{X}: \psi^*_{i-1}(\bm{x}) \geq 0 \}, i\in \{1,\dots,m\},
	\end{align}
    The key idea is to establish ISSf for sets $C^*_{i}$ with $i\in \{1,\dots,m\}$, such that the forward invariance of original safety set $C_i$ is ensured.
    Given the functions $\psi^*_i(\bm{x})$ for $i\in \{0,\!\dots\!,m\!-\!1\}$, the state-dependent admission set $K_s(\bm{x})$ of input $\bm{u}$ is defined as
	\begin{align} \label{eqn_Ks}
		K_s(\bm{x})=\Big\{ &\bm{u} \in \mathbb{U}: \frac{\partial\psi^*_{m-1}}{\partial \bm{x}} \big(\bm{f}(\bm{x})+\bm{g}(\bm{x})\bm{u}+\bm{\mu}(\bm{x}) \big)\\
		&\geq\! \frac{1}{\epsilon(\psi_{m-1}(\bm{x}))} \Big\|\frac{\partial\psi^*_{m-1}}{\partial \bm{x}} \Big\|^2 \!-\! \alpha^*_m(\psi^*_{m-1}(\bm{x}))\Big\}, \nonumber
	\end{align}
    where $\alpha^*_m: \mathbb{R} \to \mathbb{R}$ belongs to class $\mathcal{K}$, $\epsilon(\cdot): \mathbb{R} \!\to\! \mathbb{R}_{0,+}$ is a designed continuously differentiable function satisfying
    \begin{align} \label{eq: d_epsilon}
		\frac{\mathrm{d}\epsilon(\psi_{m-1}(\bm{x}))}{\mathrm{d}\psi_{m-1}} \leq 0, && \forall \psi_{m-1}(\bm{x}) \in \mathbb{R}.
	\end{align}
    Based on the admissible input set $K_s(\cdot)$, the definition of GP-based parameterized HOCBF is given as follows.
    \begin{definition} [GP-P-HOCBF] \label{definition_GPPHOCBF} 
        A continuous differentiable function $h: \mathbb{X} \to \mathbb{R}$ is a GP-based parameterized HOCBF (GP-P-HOCBF) of relative degree $m$ for system \eqref{dynamic-disturbance}, if there exist functions $\gamma_i$ and $\alpha^*_i$ for $i \in \{1,\dots,m\}$ such that $K_s(\bm{x})$ is non-empty for $\forall \bm{x} \in \mathbb{X}$.
    \end{definition}
    \cref{definition_GPPHOCBF} requires to find the suitable functions $\gamma_i$ and $\alpha^*_i$ for $i \in \{1,\dots,m\}$, which is shown later for a given $\epsilon(\cdot)$.
    Moreover, the relationship between $\psi^*_i(\bm{x})$ and original functions $\psi(\bm{x})_i$ for $i={0, \dots, m-1}$ is shown in the following theorem. For notational simplicity, the following $\psi^*_i(\bm{x})$ and $\psi_i(\bm{x})$ is denoted as $\psi^*_i$ and $\psi_i$ respectively.
    
	\begin{thm}
    \label{theorem}
        Given the system \eqref{dynamic-disturbance} satisfying \cref{assump1} and \ref{assump2} and a GP-P-HOCBF $h(\cdot)$.
        For a given function $\epsilon$ satisfying \eqref{eq: d_epsilon} and initial states $\bm{x}(t_0) \in C_1\!\cap\!\dots\!\cap\! C_m$, choose $\alpha^*_i$ for $i \in \{ 1, \dots, m\}$ as linear and satisfying \eqref{eqn_psi_star_derivative_definition} with
        \begin{align}
        \label{eq: gamma}
            \gamma_i(\psi, \bar{\eta}^2) \!=\! \alpha_{i}^{*-1} \!\circ\! \alpha_{i+1}^{*-1} \!\circ\! \dots \!\circ\! \alpha_m^{*-1} \Big(\frac{\epsilon(\psi) \bar{\eta}^2}{4} \Big).
        \end{align}
        Then, the input $\bm{u} \!\in\! K_s(\bm{x})$ for all $\bm{x} \!\in\! \mathbb{X}$ renders set $C_1\!\cap\!\dots\!\cap\! C_m$ forward invariant, indicating the safety of system \eqref{dynamic-disturbance}.
        \looseness=-1
    \end{thm}
	\begin{proof}
		As the control input $\bm{u} \in K_s(\bm{x})$ for all $\bm{x}\in \mathbb{X}$ with disturbance $\bm{d}(\bm{x})$ estimated by GP, the time derivative of $\psi^*_{m-1}$ is written as
		\begin{align}
			\label{eq: dot_psi*}
				\dot{\psi}^*_{m-1}=&\frac{\partial\psi^*_{m-1}}{\partial \bm{x}}(\bm{f}(\bm{x})+\bm{g}(\bm{x})\bm{u}+\bm{d}(\bm{x})) \nonumber \\
				\geq& \frac{\partial\psi^*_{m-1}}{\partial \bm{x}}(\bm{f}(\bm{x})+\bm{g}(\bm{x})\bm{u}+\bm{\mu}(\bm{x}) ) \nonumber \\
                &~~~~~~~~~~~~ -\Big\| \frac{\partial\psi^*_{m-1}}{\partial \bm{x}}\Big\| \Big\|\bm{d}(\bm{x})-\bm{\mu} (\bm{x}) \Big\| \\
				\geq& \frac{1}{\epsilon(\psi_{\max})} \Big\|\frac{\partial\psi^*_{m-1}}{\partial \bm{x}} \Big\|^2 -\alpha^*_m(\psi^*_{m-1}) -\Big\| \frac{\partial\psi^*_{m-1}}{\partial \bm{x}} \Big\| \bar{\eta}. \nonumber
		\end{align}
        where the second inequality applies the prediction error bound of GP in \cref{lemma: error-bound}.
        Then, by Young's inequality with a positive $\epsilon(\psi_{m-1}(\bm{x}))$, \eqref{eq: dot_psi*} becomes the following form as
			\begin{align}\label{eq: d_psi0}
				\dot{\psi}^*_{m-1}\geq& \frac{1}{\epsilon(\psi_{m-1})} \Big\| \frac{\partial\psi^*_{m-1}}{\partial \bm{x}} \Big\|^2 -\alpha^*_m(\psi^*_{m-1})  \nonumber \\
                & - \Big( \frac{1}{\epsilon(\psi_{m-1})} \Big\| \frac{\partial\psi^*_{m-1}}{\partial \bm{x}} \Big\|^2 + \frac{\epsilon(\psi_{m-1})\bar{\eta}^2}{4} \Big)  \\
				=& -\alpha^*_m(\psi^*_{m-1})-\frac{\epsilon(\psi_{m-1})\bar{\eta}^2}{4} \nonumber
			\end{align}
		
		According to the definition $\psi^*_{m-1}=\psi_{m-1}-\gamma_m(\psi_{m-1}, \bar{\eta}^2)$ and $\gamma_m(\psi_{m-1}, \bar{\eta}^2)=\alpha_m^{*-1}(\frac{\epsilon(\psi_{m-1})\bar{\eta}^2}{4})$ in \eqref{eq: gamma}, the time derivative of $\psi^*_{m-1} = \psi_{m-1} - \gamma_m(\psi_{m-1}, \bar{\eta}^2)$ satisfies
			\begin{align}
				\Big(1-\frac{\partial \gamma_m}{\partial \psi_{m-1}}\Big) \dot{\psi}_{m-1} \geq& -\alpha^*_m(\psi_{m-1}-\gamma_m(\psi_{m-1}, \bar{\eta}^2)) \nonumber \\
                &-\frac{\epsilon(\psi_{m-1})\bar{\eta}^2}{4} \\
				=& -\alpha^*_m(\psi_{m-1}) \nonumber
			\end{align}
        by using \cref{lemma: cbf}, which results in $\psi_{m-1}(\bm{x}(t)) \ge 0$ for $\forall t \ge 0$. 
        Furthermore, considering $\psi^*_{m-1} = \psi_{m-1} - \gamma_m(\psi_{m-1}, \bar{\eta}^2)$, it yields
		\begin{align}
			\label{eq: psi*}
			\psi^*_{m-1}\geq -\alpha_m^{*-1} \Big( \frac{\epsilon(\psi_{m-1})\bar{\eta}^2}{4} \Big) = - \gamma_m(\psi_{m-1}, \bar{\eta}^2).
		\end{align}    
		Based on the definition of $\psi^*_{m\!-\!1} \!\leq\! \dot{\psi}^*_{m\!-\!2} \!+\! \alpha^*_{m\!-\!1}(\psi^*_{m\!-\!2})+ \Delta_{m-1}(\bm{x})$ , it has
		\begin{align}
			\label{eq: d_psi3}
			\dot{\psi}^*_{m-2}\geq -\alpha^*_{m-1}(\psi^*_{m-2})- \gamma_m(\psi_{m-2}, \bar{\eta}^2).
		\end{align}
		Noting that (\ref{eq: d_psi3}) is similar to (\ref{eq: d_psi0}), and then the process from \eqref{eq: d_psi0} to \eqref{eq: psi*} can be used as induction step for $\psi^*_{i}(\bm{x}), i={0, \dots, m-2}$, where the known information is 
			\begin{align}
				\dot{\psi}^*_i
				&\geq -\alpha^*_{i+1}(\psi^*_{i})-\gamma_{i+2}(\psi_{i}, \bar{\eta}^2)
			\end{align}
		Similarly, applying the linearity of $\alpha^*_{i}(\cdot)$ and $\psi^*_{i}(\bm{x}) = \psi_{i}(\bm{x}) - \gamma_{i+1}(\psi_{i}(\bm{x}), \bar{\eta}^2)$, it has
		\begin{equation}
			\label{eqn_dot_psi_i}
			\begin{aligned}
				\dot{\psi}_{i}(\bm{x}) \ge& - \Big(1 - \frac{\partial \gamma_{i+1}(\psi_{i}(\bm{x}))}{\partial \psi_{i}(\bm{x})} \Big)^{-1}\Big(\alpha^*_{i+1}( \psi_{i}(\bm{x}) ) \\
				&+ \alpha^*_{i+1}( \gamma_{i+1}(\psi_{i}(\bm{x}), \bar{\eta}^2) ) - \gamma_{i+2}(\psi_{i}(\bm{x}), \bar{\eta}^2)\Big) \\
				\ge& - \Big(1 - \frac{\partial \gamma_{i+1}(\psi_{i}(\bm{x}))}{\partial \psi_{i}(\bm{x})} \Big)^{-1}\alpha^*_{i+1}( \psi_{i}(\bm{x}) ),
			\end{aligned}        
		\end{equation}
		which results in $\psi_{i}(\bm{x}(t)) \ge 0$ and $\psi^*_{i}(\bm{x}(t)) \ge - \gamma_{i+1}(\psi_{i}(\bm{x}(t)), \bar{\eta}^2)$ for $\forall t\geq t_0$, due to the fact that the negative form of the right hand side in \eqref{eqn_dot_psi_i} is an extended class $\mathcal{K}$ function.
		Combining the base case and induction step, it is straightforward to see $\bm{x}(t) \in C_1\cap\dots\cap C_m$ for $\forall t\geq t_0$, which concludes the proof.   
	\end{proof}
	
    Note that the difference between $\psi_i$ and $\psi^*_i$, i.e., $\gamma_i = \psi_i - \psi^*_i$ from \eqref{eq: psi_0}, determines the shrunk of the safe set from $C_i$ to $C^*_i$ for $i=\{0, \dots, m-1\}$.
    While \eqref{eq: psi_0} defines the shrunk starting from $\psi_0(\cdot) = h(\cdot)$, it is possible for any $i \in \{0, \dots, m-1\}$ to keep the original constraints for $j = 0, \dots, i-1$ and only to shrink the safe set after $i$, i.e.,
    \begin{align*}
		\psi^*_{j}(\bm{x})=\begin{cases}
			\psi_{j}(\bm{x}), & \text{   if } j<i \\
			\psi_{j}(\bm{x})-\gamma_{j+1}(\psi_j, \bar{\eta}^2), &\text{   otherwise}   
		\end{cases}.
	\end{align*}
    Such constructions provide more flexibility to choose which order of safe constraint reflected by $\psi_i$ to be shrunk to ensure identical original safety of the system.
    Take the second-order Euler-Lagrange system as an example with $\bm{x} = [\bm{q}^T, \dot{\bm{q}}^T]^T$ with safety constraint $h(\bm{q})$, where $\bm{q}$ and $\dot{\bm{q}}$ denote the generalized position and velocity respectively.
    When it requires a larger motion area indicating less or no shrunk of the position-related constraint $h(\bm{x})$, it is possible to add more limitations on the possible domain of the velocity $\dot{\bm{q}}$, such that the velocity-related $\psi_1$ is shrunk to compensate the effects of disturbance and guarantee the safety. 
	
	\begin{rmk}
    \label{rmk1}
		Noting that the condition on $\epsilon(\psi_{m-1})$ in (\ref{eq: d_epsilon}) is stronger than necessary. In particular, the derivation of $\epsilon(\psi_{m-1})$ only needs to satisfy that
		\begin{align}
			\label{eq: rmk3}
			\frac{\mathrm{d}\epsilon(\psi_{m-1})}{\mathrm{d}\psi_{m-1}} \leq \frac{4}{\bar{\eta}^2}\frac{1}{D(\epsilon(\psi_{m-1})\bar{\eta}^2/4)},
		\end{align}
		where $D(\epsilon(\psi_{m-1})\bar{\eta}^2/4)=\frac{\mathrm{d}\alpha_m^{-1}(\epsilon(\psi_{m-1})\bar{\eta}^2/4)}{\mathrm{d} (\epsilon(\psi_{m-1})\bar{\eta}^2/4)}$. When $\bar{\eta} \to 0$ indicating no prediction error, the right side of \eqref{eq: rmk3} approaches $\infty$, which makes $\epsilon(\cdot)$ unconstrained. 
        Therefore, the condition \eqref{eq: d_epsilon} is set. 
        Moreover, another advantage of the condition \eqref{eq: d_epsilon} is, that the overcompensation inside the set $C_i$ is prevented as the $\epsilon(\psi_{m-1})$ has a smaller value when $\psi_{m-1} \!\gg\! 0$, which effectually decreases the control performance.
        \looseness=-1
	\end{rmk}
	
	\begin{rmk}
    \label{rmk2}
		Noting that when $\frac{\mathrm{d}\epsilon}{\mathrm{d}\psi_{m-1}}= 0$, the function $\epsilon(\psi_{m-1})$ returns a positive constant. 
        As a result, a large constant imposes strong restrictions between $\psi_{m-1}$ and $\psi^*_{m-1}$, while a small constant leads to a large $\frac{1}{\epsilon(\psi_{m-1})} \|\frac{\partial\psi^*_{m-1}}{\partial \bm{x}}\|^2$ such that overcompensation and conservatism performance. 
        Therefore, selecting $\frac{\mathrm{d}\epsilon}{\mathrm{d}\psi_{m-1}}(\psi_{m-1}(\bm{x}))\neq 0$ is a more flexibility method in designing controllers. 
	\end{rmk}
	
	Based on \cref{theorem}, the following quadratic program (QP) subject to condition \eqref{eqn_Ks} is usually utilized to generate an optimization-based controller as
		\begin{align} \label{eq: qp}
			&\bm{u}^*=\mathop{\arg \min}\nolimits_{\bm{u} \in \mathbb{U}} \| \bm{u}-\bm{u}_{\text{nom}} \|^2\\
			\mathrm{s.t.} ~~&\frac{\partial\psi^*_{m-1}}{\partial \bm{x}} \big(\bm{f}(\bm{x}) \!+\! \bm{g}(\bm{x})\bm{u} \!+\! \bm{\mu}(\bm{x}) \big) \!-\! \frac{1}{\epsilon(\psi_{m-1})} \Big\|\frac{\partial\psi^*_{m-1}}{\partial \bm{x}} \Big\|^2 \nonumber \\
			&~~~~~~~~~ +\alpha^*_m(\psi^*_{m-1}(\bm{x})) \geq 0, \nonumber
		\end{align}
	where $\bm{u}_{\text{nom}} \in \mathbb{U}$ is a nominal controller.
    The control inputs $\bm{u} = \bm{u}^*$ obtained by \eqref{eq: qp} ensure the safety of system \eqref{dynamic-disturbance}.

    \begin{rmk}
    \label{rmk3}
        In \eqref{eq: qp}, the term $\frac{1}{\epsilon(\psi_{m-1})} \|\frac{\partial\psi^*_{m-1}}{\partial \bm{x}}\|^2$ can be decreased by increasing the function $\epsilon(\psi_{m-1})$. Compared to robust HOCBF via fully compensating disturbance, the proposed method in this paper is more flexible to deal with the feasibility of QP in \eqref{eq: qp}. 
        However, a large $\epsilon(\psi_{m-1})$ shrinks the safe set so much that the performance conservatism increases. 
        Therefore, the trade-off between the feasibility of QP and performance conservatism should be selected according to the practical application, which is left as future works.   
        \looseness=-1
    \end{rmk}
 
	\section{SIMULATION}
	To verify the efficacy of the proposed method, a simulation is conducted on a virtual Franka Emika manipulator with 7 degrees of freedom (DOF) in the open-source robot simulation software Coppeliasim, which closely simulates the physical robot and faithfully replicates real-world scenarios encountered in robot manipulations. 
    To provide greater specificity, the simulation task is tailored to achieve obstacle avoidance with the manipulator’s end-effector while tracking a predetermined trajectory. The obstacle is configured as a sphere, with a radius of $r=0.02m$, situated at the coordinates $[x_0, y_0, z_0]=[0.295, 0.038, 0.458]m$. 
	Regarding to the Franka Emika manipulator, the dynamic model is  
	\begin{equation}
		\label{robot}
		\bm{M}(\bm{q})\ddot{\bm{q}}+\bm{C}(\bm{q},\dot{\bm{q}})\dot{\bm{q}}+\bm{G}(\bm{q}) + \bm{d}(\bm{x}) = \bm{u},
	\end{equation}
	where $\bm{q}, \dot{\bm{q}}, \ddot{\bm{q}} \in \mathbb{R}^7$ denote the angle, velocity and acceleration of joints respectively, and $\bm{x} = [\bm{q}^T, \dot{\bm{q}}^T]^T$.
    The matrices $\bm{M}(\bm{q}) \in \mathbb{R}^{7\times 7}$, $\bm{C}(\bm{q},\dot{\bm{q}})\in \mathbb{R}^{7\times 7}$ and $\bm{G}(\bm{q})\in \mathbb{R}^{7}$ are inertia matrix, Coriolis-centrifugal matrix and gravitational term obtained from the software, respectively.
    The unknown function $\bm{d}(\bm{q}, \dot{\bm{q}}) \in \mathbb{R}^7$ is the disturbance, which is in the form $\bm{d}(\bm{x}) = 0.2 ( \bm{C}(\bm{q},\dot{\bm{q}}) \dot{\bm{q}} + \bm{G}(\bm{q})$.
	The constraints associated with this task are succinctly expressed through the continuously differentiable function $h(\bm{q}): \mathbb{R}^7 \to \mathbb{R}$ as
	\begin{equation}
		h(\bm{q})=(x(\bm{q})-x_0)^2+(y(\bm{q})-y_0)^2+(z(\bm{q})-z_0)^2-r^2,
	\end{equation}
	where $ x(\bm{q}) $, $ y(\bm{q}) $ and $ z(\bm{q}) $ are the coordinates of end-effector. Obviously, the relative degree of $h(\bm{q})$ is $m=2$. 
	
	In this simulation, the initial states of robot is $ \bm{q}(0)=[0, -\pi/4, 0,$ $ -3\pi/4, 0, 3\pi/4, \pi/4]^Trad $, and sampling time is $0.001s$. 
    The nominal controller is a PD controller as $\bm{u}_{\text{nom}}=30(\bm{q}-\bm{q_d})+15(\dot{\bm{q}}-\dot{\bm{q}}_{\bm{d}})$, where $\bm{q_d}, \dot{\bm{q}}_{\bm{d}}$ are desired joint angle and velocity respectively. 
    To compensate for the disturbance, GP is chosen with squared-exponential kernel function $k(\bm{x}, \bm{x}')=\exp(-\| \bm{x}-\bm{x}' \|^2/2l^2)$ with $l=1$. 
    The number of training data points is $M = 400$, and the boundary of measurement noise is $\sigma_v=0.02$. 
    Given the kernel function and training data, the computed upper bound of the estimation error is $\bar{\eta}^2=12.68$. 
    The extended class $\mathcal{K}$ functions are set as $\alpha_1(h)=5h$ and $\alpha_2(\psi_1)=10\psi_1$. The function $\epsilon(\psi_i)=\frac{\epsilon_0}{\mathrm{exp}(\lambda \psi_i)}, i=0, 1$ with $\epsilon_0 = 1, \lambda \in \mathbb{R}_{0,+}$ is set to construct the function $\psi^*_i$. Based on these parameter settings, the simulation result is shown as follows.

    \cref{fig1} shows the trajectories of end-effector and the curves of $h(\bm{q})$ under different control methods. When the original HOCBF without considering the prediction error is used to solve the control input, safety cannot be ensured. In contrast, the proposed GP-P-HOCBF with different functions $\epsilon(\psi_i), i=0, 1$ ensures safety and exhibits varying degrees of performance conservatism. When $\epsilon(\psi_i), i=0, 1$ is a constant, i.e., $\lambda=0$, the performance conservatism obviously increase compared to the case where $\epsilon(\psi_i), i=0, 1$ depends on the $\psi_i$, i.e., $\lambda=100$, which aligns with the concept discussed in \cref{rmk2}. This trend is further elucidated in \cref{fig2}, which shows the original safety constraint function $h(\bm{q})$ and the shrunk safety constraint function $h^*(\bm{q})$ with different $\lambda$. Regarding to GP-rHOCBF in \cite{wang2018safe}, although the performance conservatism involved in this approach is small, ensuring the feasibility of the quadratic program (QP) may be compromised in the presence of input constraints. This issue is highlighted in \cref{fig3}, where the control input of the second joint solved by GP-rHOCBF violates the input constraint. It sufficiently illustrates the proposed method is more flexible to deal with the feasibility of the QP problem. 
    \begin{figure}
	\centering
	\subfigure[]{	
		\begin{minipage}[t]{0.9\linewidth}				
			\includegraphics[width=\columnwidth]{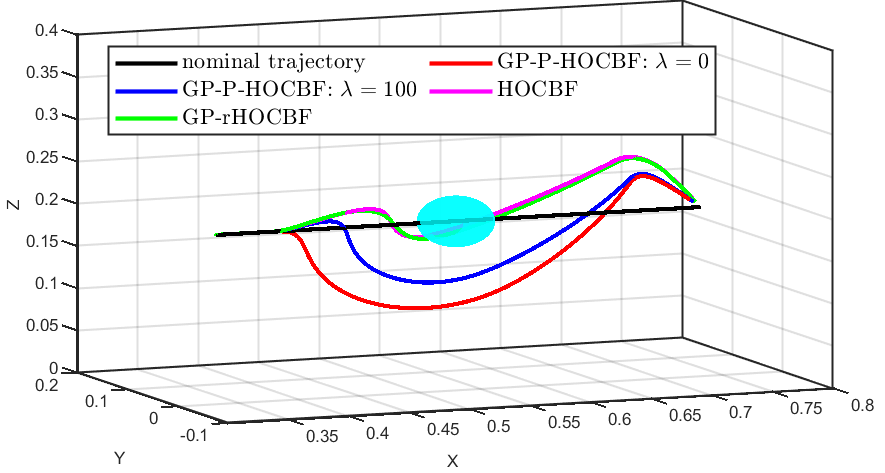}		
		\end{minipage}
	}
	\subfigure[]{
		\begin{minipage}[t]{0.9\linewidth}				
			\includegraphics[width=\columnwidth]{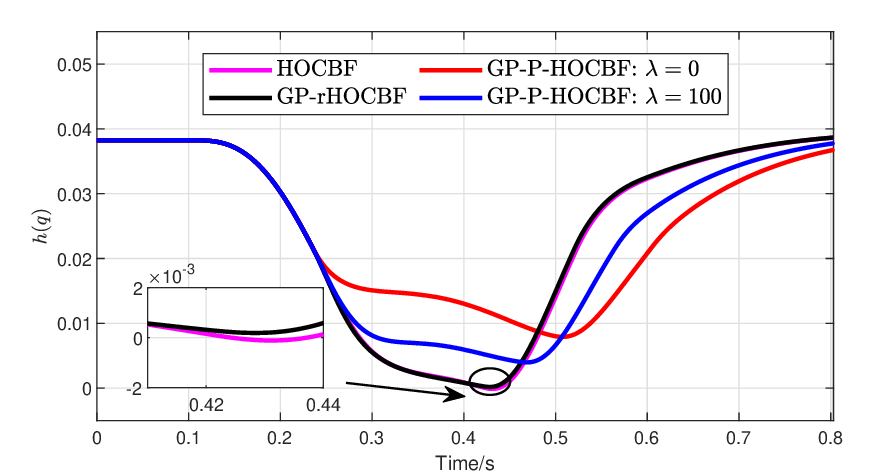}			
		\end{minipage}
	}
    \vspace{-0.3cm}
    \caption{The trajectory of end-effector and safety constraint function $h(\bm{q})$. (a)The trajectory of end-effector under different controllers. (b)The curves of $h(\bm{q})$ under different controllers.}
    \vspace{-0.3cm}
    \label{fig1}
    \end{figure}
    
    \begin{figure}
        \centering
        \includegraphics[width=\linewidth]{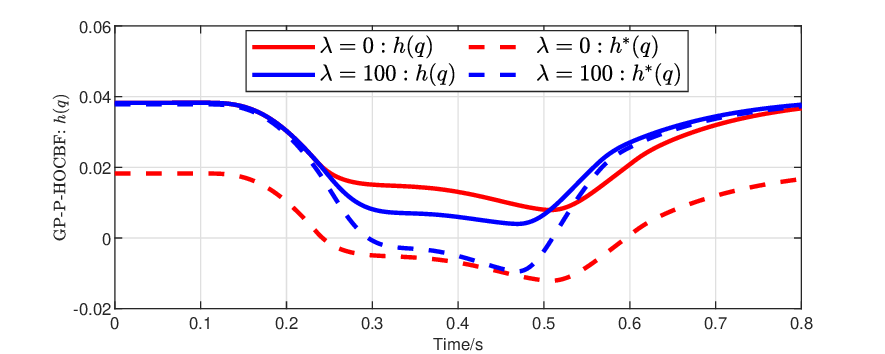}
        \vspace{-0.3cm}
        \caption{The curves of $h(\bm{q})$ and $h^*(\bm{q})$ with different $\lambda$.}
        \vspace{-0.3cm}
        \label{fig2}
    \end{figure}
	\begin{figure}
	    \centering
	    \includegraphics[width=\linewidth]{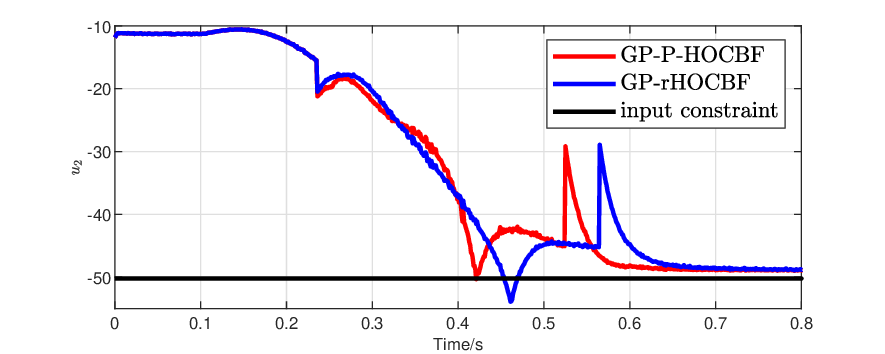}
        \vspace{-0.3cm}
	    \caption{The control input of the second joint $\bm{u}_2$ solved by different control methods.}
        \vspace{-0.3cm}
	    \label{fig3}
	\end{figure}

	\section{CONCLUSIONS}
	In this paper, the concept of Gaussian Process-based parameterized high order control barrier function (GP-P-HOCBF) is proposed to ensure the safety of systems with unknown disturbance. The disturbance is estimated by Gaussian process regression, whose prediction error characterized by the deterministic error bound is compensated by designing parameterized safe sets from the original safety constraints. The proposed GP-P-HOCBF ensures the safety of the original safety requirements, while providing more flexibility to deal with the conservatism and feasibility of solving the GP-P-HOCBF-based QP problem. Finally, the effectiveness of the proposed method is validated through simulation on the Franka Emika manipulator.



	\bibliographystyle{IEEEtran}
	\bibliography{ref}
	
	
\end{document}